\documentclass{article}
\usepackage{graphicx} 

\usepackage[utf8]{inputenc}
\usepackage{etex}
\usepackage{graphicx}
\usepackage{amscd}
\usepackage{amsfonts}
\usepackage{amssymb}
\usepackage{amsmath}
\usepackage[amsmath,amsthm,thmmarks]{ntheorem}
\usepackage{mathtools}
\usepackage{curves}
\usepackage{eucal}
\usepackage{epsfig}
\usepackage{enumerate}
\usepackage{extarrows}
\usepackage{fullpage}
\usepackage{hyperref}
\usepackage[utf8]{inputenc}
\usepackage{ifthen}
\usepackage{latexsym}
\usepackage[usenames,dvipsnames]{pstricks}
\usepackage{pst-grad} 
\usepackage{pst-plot} 
\usepackage[all]{xy}
\usepackage{tikz}
\usepackage{verbatim}
\usepackage{xargs}

\newcounter{comments}
\usepackage{comment}

 \newcommand{\new}[2][]{#2}

\DeclareMathOperator{\vecx}{\mathbf{x}}
\DeclareMathOperator{\FF}{\mathbb{F}}

\DeclareFontFamily{OT1}{rsfs}{}
\DeclareFontShape{OT1}{rsfs}{n}{it}{<-> rsfs10}{}
\DeclareMathAlphabet{\mathscr}{OT1}{rsfs}{n}{it}

\newtheorem{ex}{Exercise}

\newtheorem{dfn}{Definition}[section]

\newtheorem{cor}[dfn]{Corollary}
\newtheorem{eg}[dfn]{Example}

\newtheorem{lem}[dfn]{Lemma}

\newtheorem{remark}[dfn]{Remark}

\newtheorem{thm}[dfn]{Theorem}
\newlength{\remaining}

\newcommand{\ee}{\end{equation}}
\newcommand{\eea}{\end{eqnarray}}
\newcommand{\eean}{\end{eqnarray*}}
\newcommand{\een}{\end{equation*}}
\newcommand{\epm}{\end{pmatrix}}

\newcommand{\x}{\xymatrix@1@=50pt@M=4pt@L=3pt}

\providecommand{\keywords}[1]
{
  \small	
  \textbf{\textit{Keywords---}} #1
}

\title{\new{Modular control of Boolean network models}}

\author{David Murrugarra$^1$ \and Alan Veliz-Cuba$^2$ \and Elena Dimitrova$^3$ \and Claus Kadelka$^4$ \and Matthew Wheeler$^5$ \and Reinhard Laubenbacher$^5$}

\begin{document}

\maketitle
{\footnotesize
     \centerline{$^1$Department of Mathematics,
      University of Kentucky, Lexington, KY 40506, USA}
}
{\footnotesize
     \centerline{$^2$Department of Mathematics,
      University of Dayton, Dayton, OH 45469, USA}
}
{\footnotesize
     \centerline{$^3$Mathematics Department,
      California Polytechnic State University, San Luis Obispo, CA 93407, USA}
}
{\footnotesize
     \centerline{$^4$Department of Mathematics,
      Iowa State University, Ames, IA 50011, USA}
}
{\footnotesize
     \centerline{$^5$Department of Medicine,
      University of Florida, Gainesville, FL 32610, USA}
}

\begin{abstract}
\new{The concept of control is crucial for effectively understanding
and applying biological network models. Key structural features
relate to control functions through gene regulation, signaling, or metabolic
mechanisms, and computational models need to encode these. Applications
often focus on model-based control, such as in biomedicine or metabolic engineering.
In a recent paper, the authors developed a theoretical framework
of modularity in Boolean networks, which lead to a canonical semidirect
product decomposition of these systems. In this paper, we present an approach
to model-based control that exploits this modular structure, as well
as the canalizing features of the regulatory mechanisms. We show how to
identify control strategies from the individual modules, and we present a criterion
based on canalizing features of the regulatory rules to identify modules
that do not contribute to network control and can be excluded. For even
moderately sized networks, finding global control inputs is computationally
challenging. Our modular approach leads to an efficient approach to solving
this problem. We apply it to a published Boolean network model of blood
cancer large granular lymphocyte (T-LGL) leukemia to identify a minimal
control set that achieves a desired control objective.}
\end{abstract}


\section{Introduction}
With the availability of more experimental data and information about the structure of biological networks, computational modeling can capture increasingly complex features of biological networks~\cite{aguilar2020generalizable,plaugher2021modeling}. However, the increased size and complexity of dynamic network models also poses challenges in understanding and applying their structure as a tool for model-based control, important for a range of applications~\cite{rozum2022leveraging,plaugher2023phenotype}.
This is our focus here. To narrow the scope of the problems we address we limit ourselves to intracellular networks represented by Boolean network (BN) 
models. BNs are widely used in molecular systems biology to capture the coarse-grained dynamics of a variety of regulatory networks~\cite{schwab2020concepts}. They have been shown to provide a good approximation of the dynamics of continuous processes~\cite{veliz2012relationship}. 

For the commonly-used modeling framework of ordinary differential equations, there is a well-developed theory of optimal control, which is largely absent \new{from} other modeling frameworks, such as Boolean networks or agent-based models, both frequently used in systems biology and biomedicine. Furthermore, control inputs, in many cases, are of a binary nature, such as gene knockouts or the blocking of mechanisms. \new[For BNs, there is no readily available mathematical theory that could be used for control, leaving sampling and simulation.]{Existing Boolean network control methods do not scale well.} As networks get larger, with hundreds \cite{singh2023large} 
or even thousands of nodes~\cite{aghamiri2020automated}, \new[this leaves]{only} few computational tools \new[to]{can} identify control inputs for achieving preselected objectives, such as moving a network from one phenotype (e.g., cancer) to another (e.g., normal). One approach is to reduce the system in a way that the reduced system maintains relevant dynamical properties such as its  attractors~\cite{Veliz2011,Saadatpour2013}. This allows the control methods to be applied to the reduced system, and the same controls can then be used for the original system. 

Control targets in Boolean networks have been identified by a variety of approaches: using stable motifs~\cite{zanudo2015cell,rozum2022pystablemotifs}, feedback vertex sets~\cite{zanudo2017structure,mochizuki2013dynamics}, \new{trap spaces~\cite{cifuentes2020control,trinh2022computing}}, model checking~\cite{cifuentes2022control}, and other methods~\cite{kaminski2013minimal,samaga2010computing}. A few \new{further} approaches have \new{explicitly} used strongly connected components \new{(SCCs)} for model analysis by decomposing the wiring diagram or the state space of the network~\cite{wu2010comprehensive,jarrah2010dynamics}. \new{For example, stable motifs describe a set of network nodes and their corresponding states which are such that the nodes form a minimal SCC (e.g. a feedback loop) and their states form a partial fixed point of the Boolean network~\cite{zanudo2013effective,zanudo2015cell}. An identification of the stable motifs yields not only an efficient way to derive the network attractors but also control targets. Trap spaces, defined in~\cite{klarner2015computing} and related to stable motifs, are subspaces of the entire state space of a Boolean network that the dynamics cannot escape. Trap spaces, which can often be efficiently computed for biological Boolean models, have recently been used to simplify the identification of both network attractors and controls~\cite{cifuentes2020control,trinh2022computing}.} However, to our knowledge, none of the approaches developed thus far exploit the modular structure exhibited by many biological systems, in order to identify control strategies by focusing on one module at a time, which is the approach used in this paper. \new{This approach can substantially simplify the problem of identifying suitable controls, specifically in networks that consist of multiple decently-sized modules. For example, we showed in~\cite{kadelka2023modularity} that our modular  approach can efficiently solve the control identification problem for a 69-node pancreatic cancer model that consists of three modules of size greater than one, plus a number of single-node modules. Some of the existing control methods can often not handle reasonably large models of e.g. 70 nodes, as has been shown in~\cite{plaugher2023phenotype,borriello2021basis}. It is important to note, however, that our control approach does not guarantee finding a control of minimal size or all possible controls that achieve a given objective.}

Modularity refers to the division of the system into separate units, or modules, that each have a specific function~\cite{kashtan2005spontaneous, hartwell1999molecular}. 
Modularity is a fundamental property of biological systems that is essential for the evolution of new functions and the development of robustness~\cite{kitano2004biological,lorenz2011emergence}. In \cite{kadelka2023modularity}, we developed a mathematical theory of modularity for Boolean network models and showed that one can identify network-level control inputs at the modular level. That is, we obtain global control inputs by identifying them at the local, modular level and assembling them to global control. This enables network control for much larger networks than would otherwise be computationally unfeasible. 
It is worth noting that, although the ideas and concepts behind our approach to modularity are intuitive and natural, to our knowledge, \new[there is no well-developed]{we developed the first rigorous} mathematical theory of modular decomposition of Boolean networks\new[ and the use of modularity for the purpose of control. In this paper, we further develop this approach into a mathematical theory of biological network control. ]{, which we expand to questions related to the identification of biological network controls in this paper.}

We further propose to use another property of biological networks, represented through Boolean network models. Almost all Boolean rules that describe the dynamics of over 120 published, expert-curated biological Boolean network models have the property that they exhibit some degree of canalization~\cite{kadelka2024meta}. A Boolean function is canalizing if it has one or more variables that, when they take on a particular value, they determine the value of the function, irrespective of the values of all the other variables. As an example, any variable in a conjunctive rule (e.g., \new{$x \wedge y \wedge z$}) determines the value of the entire rule, when it takes on the value 0. We derive a criterion for Boolean network models whose Boolean functions are all canalizing, that can be used to exclude certain modules from needing to be considered for the identification of controls. 

Our approach to control via modularity is summarized in Figure~\ref{fig:control_with_modularity}. We decompose the network into its constituent modules, then apply control methods to each module to identify a control target for the entire network. We show that by combining the controls of the modules, we can control the entire network. In the last part of the paper, we present theoretical results that exploit the canalizing properties of the regulatory functions to exclude certain modules from the control search. Finally, we demonstrate our approach by applying it to a published model of the blood cancer large granular lymphocyte (T-LGL) leukemia~\cite{Saadatpour:2011aa}.

\begin{figure}[t]
  \centering
   \includegraphics[width=\textwidth]{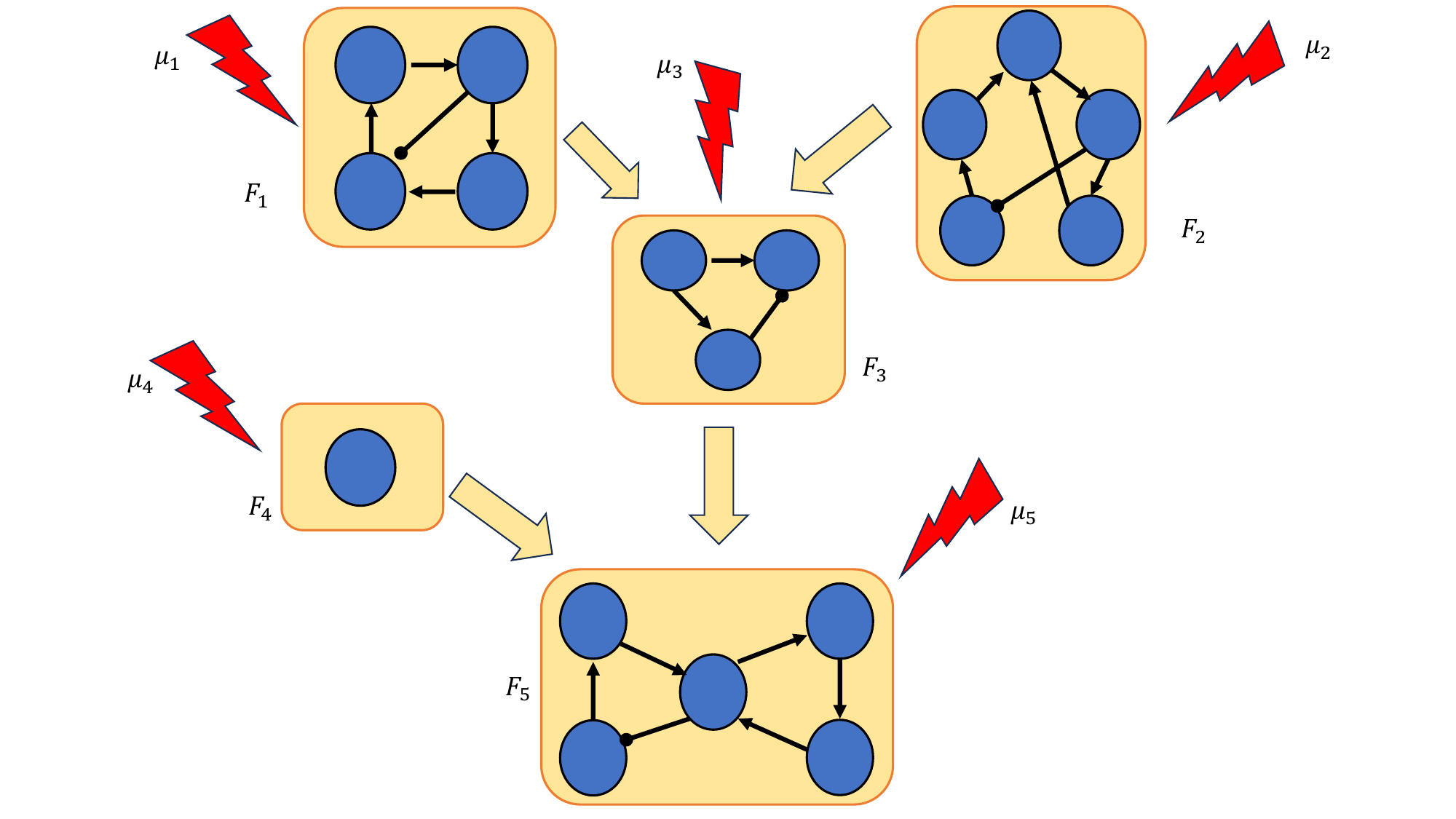}
   \caption{Control via modularity. First, the network is decomposed into its constituent modules: $F_1, \dots, F_m$. Then, controls $\mu_1, \dots, \mu_m$ are identified for each module. Combining the controls of the modules $\mu = (\mu_1,\ldots,\mu_m)$ yields a set of controls for the whole network.}
   \label{fig:control_with_modularity}
\end{figure}

\section{Background}

We first describe Boolean networks and how to decompose a network into modules.
In a BN, each gene is represented by a node that can be in one of two states: ON or OFF. Time is discretized as well, and the state of a gene at the next time step is determined by a Boolean function that takes as input the current states of a subset of the nodes in the BN. The dependence of a gene on the state of another gene can be graphically represented by a directed edge, and the \emph{wiring diagram} contains all such dependencies. 

\subsection{Boolean Networks}
Boolean networks can be seen as discrete dynamical systems. 
Specifically, consider $n$ variables $x_1,\dots,x_n$ each of which can take values in $\FF_2:=\{0,1\}$, where $\FF_2$ is the field with two elements, 0 and 1, where arithmetic is performed modulo 2. 
Then, a synchronously updated Boolean network is a function
$F=(f_1,\ldots,f_n):\FF_2^n\to\FF_2^n$, where each coordinate function $f_i$ describes how the future value of variable $x_i$ depends on the present values of all variables. All variables are updated at the same time (synchronously). 

\begin{dfn}\label{def_wiring_diagram}
The \emph{wiring diagram} of a Boolean network $F=(f_1,\ldots,f_n): \FF_2^n\rightarrow \FF_2^n$ is the directed graph with vertices $x_1,\ldots,x_n$ and an edge from $x_i$ to $x_j$ if $f_j$ depends on $x_i$. That is, if there exists $\mathbf x \in \FF_2^n$ such that 
$f_j(\mathbf x)\neq f_j(\bf x + e_i),$ where $\bf e_i$ is the $i$th unit vector.
\end{dfn}

\begin{dfn}\label{def_strongly_connected}
The wiring diagram of a Boolean network is \emph{strongly connected} if every pair of nodes is connected by a directed path. That is, for each pair of nodes $x_i,x_j$ in the wiring diagram  with $x_i\neq x_j$ there exists a directed path from $x_i$ to $x_j$ (and vice versa). In particular, a one-node wiring diagram is strongly connected by definition.
\end{dfn}

\begin{remark}
The wiring diagram of any Boolean network is either strongly connected or it is composed of a collection of strongly connected components where connections between different components 
move in only one direction.  
\end{remark}

\begin{eg}\label{eg:wd}
Figure~\ref{fig:dynamics_eg}a shows the wiring diagram of the Boolean network $F:\FF_2^3\rightarrow \FF_2^3$ given by $$F(x_1,x_2,x_3)=(x_2 \wedge \neg x_3,x_3,\neg x_1 \wedge x_2).$$
\end{eg}

\begin{figure}
    \centering
    \includegraphics[width=\textwidth]{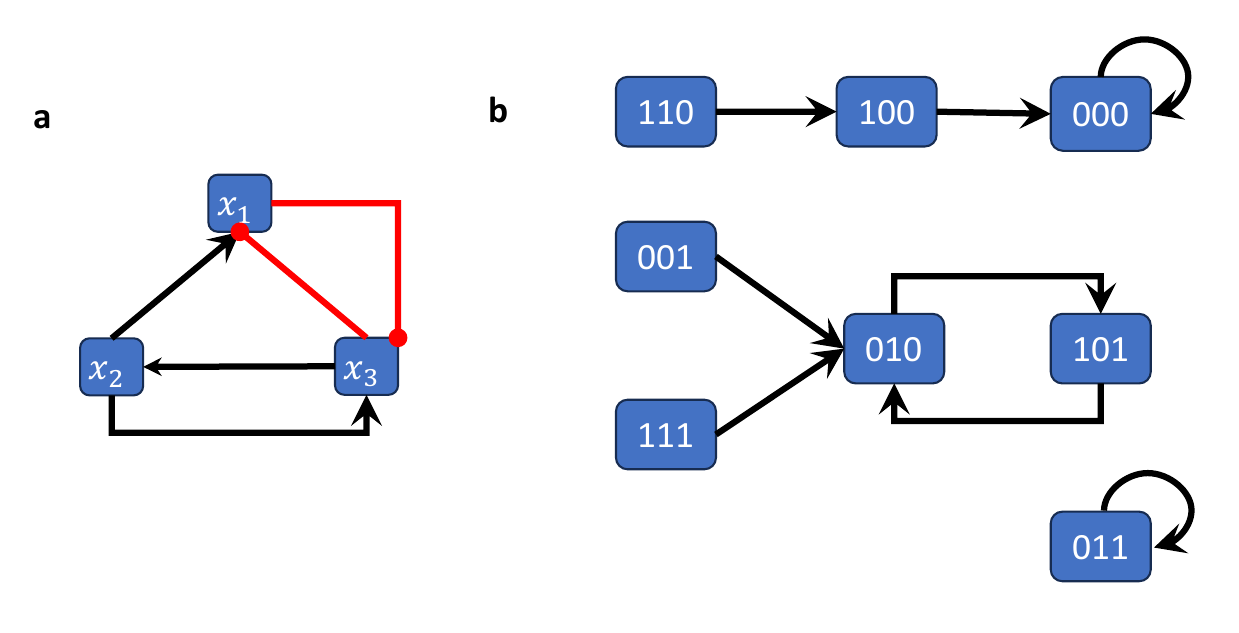}
    \caption{Wiring diagram and state space of the Boolean network in Example~\ref{eg:wd}-\ref{eg:cycle_structure}. 
    (a) The wiring diagram encodes the dependency between variables.
    (b) The state space is a directed graph with edges between all states and their images. This graph therefore encodes all possible trajectories. 
    }
    \label{fig:dynamics_eg}
\end{figure} 

\subsection{Dynamics of Boolean networks}
Another directed graph associated with a BN is the \emph{state transition graph}, also referred to as the \emph{state space}. It describes all possible transition\new{s} of the BN from one time step to \new{the next}. The \emph{attractors} of a BN are minimal sets of states from which there is no escape as the system evolves. An attractor with a single state is also called a \emph{steady state} (or fixed point). In mathematical models of intracellular regulatory networks, the attractors of the model are often associated with the possible phenotypes of the cell. This idea can be traced back to Waddington~\cite{waddington2014strategy} and Kauffman~\cite{Kauffman:1969aa}. For example, in a model of cancer cells, the steady states of the model correspond to proliferative, apoptotic, or growth-arrest phenotypes~\cite{plaugher2022uncovering}. 
Mathematically, a phenotype is associated with a group of attractors where a subset of the system’s variables have the same states. 
These shared states are then used as biomarkers that indicate diverse hallmarks of the system.

There are two ways to describe the dynamics of a Boolean network $F: \FF_2^n\to \FF_2^n$, (i) as trajectories for all $2^n$ possible  initial conditions, or (ii) as a directed graph with nodes in $\FF_2^n$. Although the first description is less compact, it will allow us to formalize the dynamics of coupled networks. 

\begin{dfn}
A trajectory of a Boolean network $F:\FF_2^n\rightarrow \FF_2^n$ is a sequence $(x(t))_{t=0}^\infty$ of elements of $\FF_2^n$ such that $x(t+1)=F(x(t))$ for all $t\geq 0$.
\end{dfn}

\begin{eg}\label{eg:traj}
For the network in the example above, $F(x_1,x_2,x_3)=(x_2 \wedge \neg x_3,x_3,\neg x_1 \wedge x_2)$, there are $2^3=8$ possible initial states giving rise to the following trajectories (commas and parenthesis for states are omitted for brevity).
\begin{align*}
     T_1&=(000,000,000,000,\ldots)\\
     T_2&=(001,010,101,010,\ldots)\\
    T_3&=(010,101,010,101,\ldots)\\
    T_4&=(011,011,011,011,\ldots)\\
    T_5&=(100,000,000,000,\ldots)\\
    T_6&=(101,010,101,010,\ldots)\\
    T_7&=(110,100,000,000,\ldots)\\
    T_8&=(111,010,101,010,\ldots)
\end{align*}
We can see that $T_3$ and $T_6$ are periodic trajectories with period 2. Similarly, $T_1$ and $T_4$ are periodic with period 1. All other trajectories eventually reach one of these 4 states.
\end{eg}

When seen as trajectories, $T_3$ and $T_6$ are different, but they can both be encoded by the fact that $F(0,1,0)=(1,0,1)$ and $F(1,0,1)=(0,1,0)$. Similarly, $T_1$ and $T_4$ can be encoded by the equalities $F(0,1,1)=(0,1,1)$ and $F(0,0,0)=(0,0,0)$. This alternative, more compact way of encoding the dynamics of a Boolean network is the standard approach, which we formalize next. 

\begin{dfn}
The \emph{state space} of a (synchronously updated) Boolean network $F:\FF_2^n\to \FF_2^n$ 
is a directed graph with vertices in $\FF_2^n$ and an edge from $x$ to $y$ if $F(x)=y$.
\end{dfn}

\begin{eg}\label{eg:state_space}
Figure~\ref{fig:dynamics_eg}b shows the state space of the (synchronously updated) Boolean network from Example~\ref{eg:wd}.
\end{eg} 

From the state space, one can easily obtain all periodic points, which form the attractors of the network. 

\begin{dfn}\label{def_attractors}
The \emph{\new{set} of attractors} for a Boolean network is the set $\mathcal{A}(F)$ of all \emph{minimal} \new{nonempty} subsets $\mathcal{C}\subseteq \FF_2^n$ satisfying $F(\mathcal{C})=\mathcal{C}$.  
\begin{enumerate}
    \item The subset $\mathcal{A}^1(F)\subset\mathcal{A}(F)$ of sets of exact size 1 consists of all \emph{steady states} (also known as \emph{fixed points}) of $F$.
    \item The subset $\mathcal{A}^r(F)\subset\mathcal{A}(F)$ of sets of exact size $r$ consists of all \emph{cycles} of length $r$ of $F$.
\end{enumerate}

Equivalently, an \emph{attractor of length $r$} is an ordered set with $r$ elements, $\mathcal{C}=\{c_1,\ldots,c_r\}$, such that $F(c_1)=c_2, F(c_2)=c_3,\ldots, F(c_{r-1})=c_r, F(c_r)=c_1$.
\end{dfn}

\begin{remark}
In the case of steady states, the attractor $\mathcal{C}=\{c\}$ may be denoted simply by $c$.
\end{remark}

\begin{eg}\label{eg:cycle_structure}
The Boolean network from Example~\ref{eg:wd} 
has 2 steady states (i.e., attractors of length 1) and one cycle of length 2, which can be easily derived from its state space representation (Figure~\ref{fig:dynamics_eg}b).
\end{eg}

\subsection{Modules}\label{subsec:simple_networks}

In~\cite{kadelka2023modularity}, a concept of modularity was introduced for Boolean networks. The decomposition into modules occurs on structural (wiring diagram) level but induces an analogous decomposition of the network dynamics, in the sense that one can recover the dynamics of the entire network from the dynamics of the modules. 
For this decomposition, a \textit{module} of a BN is defined as a subnetwork that is itself a BN with external parameters
in the subset of variables that specifies a strongly connected component (SCC) in the wiring diagram (see Example~\ref{eg:modules_eg}).
More precisely, for a Boolean network $F$ and  subset of its variables $S$, we define 
the \textit{restriction} of $F$ to $S$ to be the BN $F|_S = (f_{k_1},\new{\ldots},f_{k_m})$, where $x_{k_i}\in S$ for $i=1,\new{\ldots},m$. We note that $f_{k_i}$ might contain inputs that are not part of $S$ (e.g., when $x_{k_i}$ is regulated by some variables that are not in $S$). Therefore, the BN $F|_S$ may contain external parameters (which are themselves fixed and do not possess an update rule). Given a $F$ with wiring diagram $W$, let $W_1,\ldots ,W_m$ be the SCCs of $W$ with pairwise disjoint sets of variables $S_i$.
The \textit{{modules}} of $F$ are then the restrictions to these sets of variables, $F|_{S_i}$. Further, the modular structure of $F$ can be described by a directed acyclic graph \new{$Q =(V,E)$ with vertex set $V=\{W_1,\ldots ,W_m\}$ (where we contracted each SCC into a single vertex) and edge set $E=\{(i,j)\mid W_i\longrightarrow W_j\}$ by setting $W_i\longrightarrow W_j$ whenever there exists a node from $W_i$ to $W_j$ (we show a modular decomposition of a small BN in Example~\ref{eg:modules_eg}).} 

\begin{eg}\label{eg:modules_eg}
Consider the Boolean network 
\[F(x)=(x_2\wedge x_1, \neg x_1, x_1\vee \neg x_4, (x_1\wedge \neg x_2)\vee (x_3\wedge x_4))\] with wiring diagram in Figure~\ref{fig:modules_eg}a. The restriction of this network to $S_1=\{x_1,x_2\}$ is the 2-variable network $F|_{S_1}(x_1,x_2)=(x_2\wedge x_1, \neg x_1)$, which forms the first module (indicated by the amber box in  Figure~\ref{fig:modules_eg}a), while the restriction of $F$ to $S_2=\{x_3,x_4\}$ is the 2-variable network with external parameters $x_1$ and $x_2$ $F|_{S_2}(x_3,x_4)=(x_1\vee \neg x_4, (x_1\wedge \neg x_2)\vee (x_3\wedge x_4)) )$, which forms the second module (indicated by the green module in Figure~\ref{fig:modules_eg}a). 
Note that the module $F|_{S_1}$, i.e., the restriction of $F$ to $S_1$, is simply the projection of $F$ onto the variables $S_1$ because $W_1$ does not receive feedback from the other component.

The wiring diagram of this network has two strongly connected components $W_1$ and $W_2$ with variables $S_1=\{x_1,x_2\}$ and $S_2=\{x_3,x_4\}$ (Figure~\ref{fig:modules_eg}a), connected according to the directed acyclic graph 
\new{$Q =(V,E)$ with vertex set $V=\{W_1,W_2\}$ (where we contracted each SCC into a single vertex) and edge set 
$E = \{(1,2)\mid W_1\longrightarrow W_2\}$ (Figure~\ref{fig:modules_eg}b).} 
\end{eg}

\begin{figure}
    \centering
    \includegraphics[width=\textwidth]{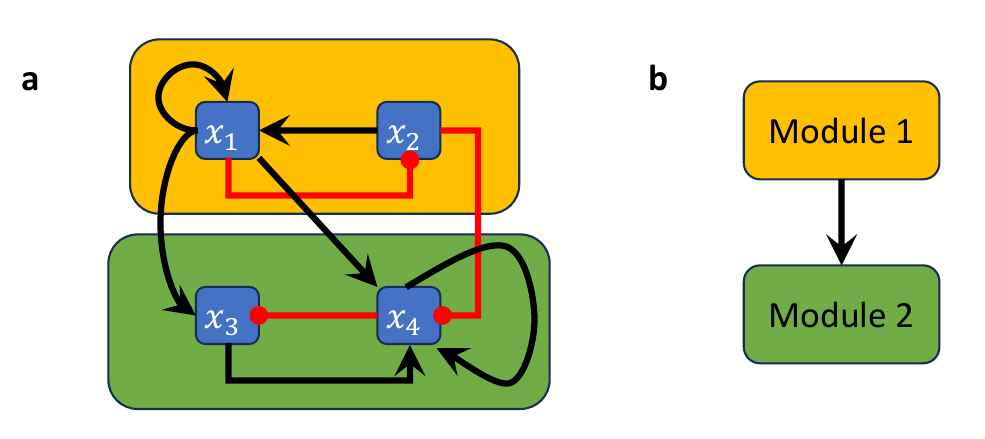}
    \caption{Boolean network decomposition into modules. (a) Wiring diagram of a non-strongly connected Boolean network where the \new[non-trivial]{} modules are highlighted by amber and green boxes. (b) Directed acyclic graph describing the corresponding connections between the \new[non-trivial]{} modules.
    }
    \label{fig:modules_eg}
\end{figure} 

\section{Control via Modularity}\label{sec:control}

In this section, we apply the modular decomposition theory described in the previous section and in \cite{kadelka2023modularity} to make the control problem of Boolean networks more tractable. We show how the decomposition into modules can be used to obtain controls for each module, which can then be combined to obtain a control for the entire network. In this context, two types of control actions are generally considered: edge controls and node controls. For each type of control, one can consider deletions or constant expressions as defined below. The motivation for considering these control actions is that they represent the common interventions that can be implemented in practice. For instance, edge deletions can be achieved by the use of therapeutic drugs that target specific gene interactions, whereas node deletions represent the blocking of effects of products of genes associated to these nodes; see~\cite{choi2012attractor,wooten2021mathematical}.

Once the modules have been identified, different methods for phenotype control (that is, control of the attractor space) can be used to identify controls in these networks. Some of these methods employ stable motifs~\cite{zanudo2015cell}, feedback vertex sets~\cite{zanudo2017structure}, as well as algebraic approaches~\cite{murrugarra2016identification,murrugarra2015molecular,sordo2020control}. For our examples below, we will use the methods defined in~\cite{zanudo2015cell, murrugarra2016identification, zanudo2017structure} to find controls for the \new{modules}. 

A Boolean network $F=(f_1,\ldots,f_n):\FF_2^n\rightarrow \FF_2^n$ with \emph{control} is 
a Boolean network $F^\mu:\mathbb{F}_2^n\times U\rightarrow \mathbb{F}_2^n$,
where $U$ is a set that denotes all possible controls, defined below. The case of no control coincides with the original Boolean network, that is, $F^\mu(x,0)=F(x)$. Given a control $\mu\in U$, the dynamics are given by $x(t+1)=F^\mu(x(t),\mu)$. See~\cite{murrugarra2016identification} for additional details and examples of how to encode control edges and nodes in a Boolean network.

\begin{dfn}[Edge Control]\label{def:edge_del}
Consider the edge $x_i\rightarrow x_j$ in the wiring diagram ${W}$. 
The function
\begin{equation}
\label{edge_del_def}
F^\mu_j(x,\mu_{i,j}) := f_j(x_1,\dots,(\mu_{i,j}+1)x_i+\mu_{i,j}a_i,\dots,x_n),
\end{equation}
where $a_i$ is a constant in $\FF_2$, encodes the control of the edge $x_i\rightarrow x_j$, since for each possible value of $\mu_{i,j}\in \mathbb{F}_2$ we have the following control settings: 
\begin{itemize}
  \item If $\mu_{i,j}=0$, $F^\mu_j(x,0) = f_j(x_1,\dots,x_i,\dots,x_n)$. That is, the control is not active.
  \item If $\mu_{i,j}=1$, $F^\mu_j(x,1) = f_j(x_1,\dots,x_i=a_i,\dots,x_n)$. In this case, the control is active, and the action represents the removal of the edge $x_i\rightarrow x_j$ when $a_i=0$, and the constant expression of the edge if $a_i=1$. We use  $x_i\xrightarrow[]{a_i} x_j$ to denote that the control is active.
\end{itemize}

\end{dfn}
This definition can be easily extended for the control of many edges, so that we obtain $F^\mu:\mathbb{F}_2^n\times \mathbb{F}_2^e \rightarrow \mathbb{F}_2^n$, where $e$ is the number of edges in the wiring diagram. Each coordinate, $\mu_{i,j}$, of $\mu$ in $F^\mu(x,\mu)$ encodes the control of an edge $x_i\rightarrow x_j$.
\begin{dfn}[Node Control]\label{def:node_del}
Consider the node $x_i$ in the wiring diagram ${W}$. The function
\begin{equation}
\label{node_del_def}
F^\mu_j(x,\mu^{-}_i,\mu^{+}_i) := (\mu^{-}_i+\mu^{+}_i+1)f_j(x) + \mu^{+}_i
\end{equation}
encodes the control (knock-out or constant expression) of the node $x_i$, since for each possible value of $(\mu^{-}_i,\mu^{+}_i)\in \mathbb{F}_2^2$ we have the following control settings:  
\begin{itemize}
  \item For $\mu^{-}_i=0, \mu^{+}_i=0$, $F^\mu_j(x,0,0) = f_j(x)$. That is, the control is not active.
  \item For $\mu^{-}_i=1, \mu^{+}_i=0$, $F^\mu_j(x,1,0)  = 0$. This action represents the knock-out of the node $x_j$.
  \item For $\mu^{-}_i=0, \mu^{+}_i=1$, $F^\mu_j(x,0,1)  = 1$. This action represents the constant expression of the node $x_j$.
  \item For $\mu^{-}_i=1, \mu^{+}_i=1$, $F^\mu_j(x,1,1) =  f_j(x_{t_1},\dots,x_{t_m})+1$. This action changes the Boolean function to its negative value. This case is usually not considered in the control search since it is biologically impractical to implement.
\end{itemize}
\end{dfn}

We note that the algebraic framework is versatile enough that we can encode any type of control, such as a combination of node and edge control at the same time.

\begin{dfn}
A control $\mu$ \emph{stabilizes} a given Boolean network $F: \FF_2^n\rightarrow \FF_2^n$ at an attractor $\mathcal{C}$ $\subseteq \FF_2^n$
when the resulting network after applying $\mu$ to $F$ (denoted here as $F^\mu$) has $\mathcal{C}$ as its only attractor. \new{Note that $\mu$ may contain a combination of one or multiple edge or node controls, as described in Definitions~\ref{edge_del_def} and \ref{def:node_del}.}
\end{dfn}

For a Boolean network $F$, we let $\mathcal{A}(F)$ denote the set of its attractors.
Whenever a Boolean network $F$ has more than one module we say that it is \textit{decomposable} into its constituent modules $F_1, F_2, \new{\ldots}, F_m$ ($m \geq 2$), and write $F=F_1\rtimes_{P_1} F_2\rtimes_{P_2}\cdots\rtimes_{P_{m-1}}F_m$ where the semi-product operation $\rtimes_{P_i}$ indicates the coupling of the subnetworks, as described in~\cite{kadelka2023modularity} (see Example~\ref{eg:coupling} for an example of coupling).
Furthermore, from the decomposition theory described in~\cite{kadelka2023modularity}, the attractors of $F$ are of the form $\mathcal{C} = \mathcal{C}_1\oplus\mathcal{C}_2\oplus\cdots\oplus\mathcal{C}_n$ where $\mathcal{C}_i\in\mathcal{A}(F_i)$ is an attractor of the subnetwork, for $i=1,\dots,n$
\new{(one can think of this direct sum operation as concatenating the attractors from the different modules, see Example~\ref{ex:control0} where we show an example of this operation)}.

\begin{eg}\label{eg:coupling}
Consider the Boolean network in Example~\ref{eg:modules_eg},
\[F(x)=(x_2\wedge x_1, \neg x_1, x_1\vee \neg x_4, (x_1\wedge \neg x_2)\vee (x_3\wedge x_4))\] with wiring diagram in Figure~\ref{fig:modules_eg}a. 
Let $F_1 = F|_{S_1}(x_1,x_2)=(x_2\wedge x_1, \neg x_1)$ be the restriction of $F$ to $S_1=\{x_1,x_2\}$ which forms the first module (indicated by the amber box in  Figure~\ref{fig:modules_eg}a) and let $F_2 = F|_{S_2}(x_3,x_4)=(e_1\vee \neg x_4, (e_1\wedge \neg e_2)\vee (x_3\wedge x_4)) )$
be the restriction of $F$ to $S_2=\{x_3,x_4\}$, which forms the second module (indicated by the green box in Figure~\ref{fig:modules_eg}a). 
Here, $e_1$ and $e_2$ are external parameters of $F|_{S_2}$ that take the place of $x_1$ and $x_2$, which is indicated by the coupling $P=\{x_1\to e_1, x_2\to e_2\}$.
$F$ is thus decomposable into $F_1$ and $F_2$ and we have $F = F_1\rtimes_P F_2$.

\end{eg}

The following theorem takes advantage of the modular structure of the network to find controls one module at a time. 

\begin{thm}\label{thm:control0}
Given a decomposable network $F = F_1\rtimes_P F_2$, if $\mu_1$ is a control that stabilizes $F_1$ in $\mathcal{C}_1$ (whether $\mathcal{C}_1$ is an existing attractor of $F_1$ or a new one, after applying control $\mu_1$) and $\mu_2$ is a control that stabilizes $F_2$ in $\mathcal{C}_2$ (whether $\mathcal{C}_2$ is an existing attractor of $F_2$ or a new one, after applying control $\mu_2$), then $\mu = (\mu_1, \mu_2)$ is a control that stabilizes $F$ 
in $\mathcal{C} = \mathcal{C}_1\oplus\mathcal{C}_2$ provided that at least one of $\mathcal{C}_1$ or $\mathcal{C}_2$ is a steady state.
\end{thm}
\begin{proof}
Let $F^{\mu_1}_1$ be the resulting network after applying the control $\mu_1$. Thus, the dynamics of $F^{\mu_1}_1$ is $\mathcal{C}_1$, that is $\mathcal{A}(F^{\mu_1}_1) =\mathcal{C}_1$. Similarly, the 
dynamics of $F^{\mu_2}_2$ 
is $\mathcal{C}_2$. 
\new{Let $F^{\mathcal{C}_1, \mu_2}_2$ be the network that results after setting all external parameters of $F_2$ to the states of $\mathcal{C}_1$ and
applying the control $\mu_2$.}
Then, $\mathcal{A}(F^{\mathcal{C}_1, \mu_2 }_2) = \mathcal{C}_2$. Thus,
\[
F^\mu = (F_1\rtimes_P F_2)^\mu = F_1^\mu\rtimes_P F_2^\mu = F_1^{\mu_1}\rtimes_P F_2^{\mu_2}.
\]
Therefore,
\[
\mathcal{A}(F^\mu) = \mathcal{A}(F_1^{\mu_1}\rtimes_P F_2^{\mu_2}) = \bigsqcup_{\mathcal{C}'\in\mathcal{A}(F_1^{\mu_1})}\mathcal{C}'\oplus\mathcal{A}(F^{\mathcal{C}',\mu_2}_2)
= \mathcal{C}_1\oplus\mathcal{A}(F^{\mathcal{C}_1,\mu_2}_2) = \mathcal{C}_1 \oplus \mathcal{C}_2.
\]
For the last equality we used the fact that the product of a steady state and a cycle (or vice versa) will result in only one attractor for the combined network. \new{Notice, however, that} multiplying two attractors (of length greater than 1) might result in several attractors for the composed network due to the attractors starting at different states. 

It follows that there is only one attractor of $F^\mu$ and that attractor is $\mathcal{C}_1\oplus\mathcal{C}_2$.
Thus, $F$ is stabilized by $\mu=(\mu_1, \mu_2)$ and we have $\mathcal{A}(F^{\mu}) = \mathcal{C}$.
\end{proof}

\begin{eg}\label{ex:control0}
Consider the network $F(x_1,x_2,x_3,x_4)=(x_2,x_1,x_2\wedge x_4,x_3)$ \new{whose wiring diagram and state space are given in Figure~\ref{fig:state-space}, which can be decomposed into $F=F_1\rtimes_P F_2$, with $F_1 = F|_{S_1}(x_1,x_2)=(x_2,x_1)$ where $S_1=\{x_1,x_2\}$ and 
$F_2=F|_{S_2}(x_3,x_4)=(e_2\wedge x_4,x_3)$ where $S_2=\{x_3,x_4\}$ and $e_2$ is an external parameter.} 
Here the coupling is given by $P=\{x_2\to e_2\}$.
Suppose we want to stabilize $F$ in 0110\new{,} which is not an attractor of $F$ (\new{see Figure~\ref{fig:state-space}(c)}).
\new{The set of attractors of $F_1$ is given by } $\mathcal{A}(F_1)=\left\{00,11,\{01,10\}\right\}$.
\begin{itemize}
\item Consider the control $\mu_1:(x_1\xrightarrow[]{1} x_2,x_2\xrightarrow[]{0} x_1)$. That is, the control is the combined action of setting the input from $x_1$ to $x_2$ to 1 and the input from $x_2$ to $x_1$ to 0. The control $\mu_1$ stabilizes $F_1$ at
$01$, 
which is not an original attractor of $F_1$. Let $\mathcal{C}_1=\{01\}\in\mathcal{A}(F_1^{\mu_1})$.
\new{Let $F^{\mathcal{C}_1}_2$ be the network that results after setting all external parameters of $F_2$ to the states of $\mathcal{C}_1$.}
Note that the space of attractors for $F_2^{\mathcal{C}_1}$ is $\mathcal{A}(F_2^{\mathcal{C}_1})=\{00,11,\{01,10\}\}$.

\item Now consider the control $\mu_2:(x_4\xrightarrow[]{1} x_3,x_3\xrightarrow[]{0} x_4)$. That is, the control is the combined action of setting the input from $x_4$ to $x_3$ to 1 and the input from $x_3$ to $x_4$ to 0. This control stabilizes $F_2^{\mathcal{C}_1}$ at $\mathcal{C}_2=\{10\}\in\mathcal{A}(F_2^{\mathcal{C}_1})$, which is not an original attractor of $F_2^{\mathcal{C}_1}$.

\item Finally, the control $\mu=(\mu_1, \mu_2)$ stabilizes $F$ at $\mathcal{C} = \mathcal{C}_1\oplus\mathcal{C}_2 = \{0110\}$ \new{by Theorem~\ref{thm:control0}}. 
\end{itemize}
\end{eg}

Theorem~\ref{thm:control0} shows how the modular structure can be used to identify controls that stabilize the network in any desired state. 
In particular, we can use the modular structure of a network to find controls that stabilize a network at an existing attractor, which is often the case in biological control applications.
We state this fact in the following corollary.

\begin{cor}\label{cor:control1}
Given a decomposable network $F = F_1\rtimes_P F_2$, let $\mathcal{C} = \mathcal{C}_1\oplus\mathcal{C}_2$ be an attractor of $F$, where $\mathcal{C}_1\in\mathcal{A}(F_1)$ and $\mathcal{C}_2\in\mathcal{A}(F^{\mathcal{C}_1}_2)$ and at least
$\mathcal{C}_1$ or $\mathcal{C}_2$ is a steady state. If $\mu_1$ is a control that stabilizes $F_1$ in $\mathcal{C}_1$ and $\mu_2$ is a control that stabilizes $F^{\mathcal{C}_1}_2$ in $\mathcal{C}_2$, then $\mu=(\mu_1,\mu_2)$ is a control that stabilizes $F$ in $\mathcal{C}$.
\end{cor}

Theorem~\ref{thm:control0} uses the modular structure of a Boolean network to identify controls that stabilize the network in any desired attractor. In biological applications, the attractors typically correspond to distinct biological phenotypes (defined more rigorously in the next section) and a common question is how to force a network to always transition to only one of these phenotypes. For example, cancer biologists may use an appropriate Boolean network model with the two phenotypes proliferation and apoptosis to identify drug targets (i.e., edge or node controls), which force the system to always undergo apoptosis. 
\new{In \cite{kadelka2023modularity}, we applied the approach described in Corollary~\ref{cor:control1} to find an efficient intervention that solves the control problem for the pancreatic cancer model with 69 nodes. 
In the following example, we illustrate the approach described in Corollary~\ref{cor:control1} using a toy model.
}

\begin{eg}\label{ex:stabilize}
Consider again the network $F(x_1,x_2,x_3,x_4)=(x_2,x_1,x_2  \wedge x_4,x_3) = F_1\rtimes_P F_2$ from Example~\ref{ex:control0} with 
\new{$F_1 = F|_{S_1}(x_1,x_2)=(x_2,x_1)$ where $S_1=\{x_1,x_2\}$ and 
$F_2=F|_{S_2}(x_3,x_4)=(e_2\wedge x_4,x_3)$ where $S_2=\{x_3,x_4\}$ and $e_2$ is an external parameter.
Here the coupling is given by $P=\{x_2\to e_2\}$.}
Suppose we want to stabilize $F$ in 1111, which is \new{one of the attractors} of $F$ (\new{see Figure~\ref{fig:state-space}(c)}). 
\new{The set of attractors of $F_1$ is } $\mathcal{A}(F_1)=\left\{00,11,\{01,10\}\right\}$. Let $\mathcal{C}_1=\{11\}\in\mathcal{A}(F_1)$.
\begin{itemize}
\item The edge control $\mu_1:x_1\xrightarrow[]{1} x_2$ (that is, the control that constantly expresses the edge from $x_1$ to $x_2$) stabilizes $F_1$ at
$\mathcal{C}_1=\{11\}$. 
\new{Let $F^{\mathcal{C}_1}_2$ be the network that results after setting all external parameters of $F_2$ to the states of $\mathcal{C}_1$.}
The space of attractors for $F_2^{\mathcal{C}_1}$ is then $\mathcal{A}(F_2^{\mathcal{C}_1})=\{00,11,\{01,10\}\}$. Note that $x_2\xrightarrow[]{1} x_1$ would be an alternative control. 

\item The edge control $\mu_2:x_4\xrightarrow[]{1} x_3$ (that is, the control that constantly expresses the edge from $x_4$ to $x_3$) stabilizes $F_2^{\mathcal{C}_1}$ at $\mathcal{C}_2=\{11\}\in\mathcal{A}(F_2^{\mathcal{C}_1})$. Again, note that $x_3\xrightarrow[]{1} x_4$ would be an alternative control.

\item Now, the control $\mu=(\mu_1,\mu_2) = (x_1\xrightarrow[]{1} x_2, x_4\xrightarrow[]{1} x_3)$ stabilizes $F$ at $\mathcal{C} = \mathcal{C}_1\oplus\mathcal{C}_2 = \{1111\}$ \new{by Corollary~\ref{cor:control1}}.
\end{itemize}
\end{eg}

\begin{remark}\label{rem:control0}
In Theorem~\ref{thm:control0}, we required  one of the stabilized attractors to be a steady state in order to be able to combine the controls from the modules.
We can remove this requirement by using the following definition of stabilization for non-autonomous networks \new{(defined in~\cite{kadelka2023modularity})},
which will guarantee that $\mathcal{C}_1$ and $\mathcal{C}_2$ can be combined in a unique way, resulting in a unique attractor of the whole network. 
\end{remark}

\begin{dfn}\label{def:non_aut}\cite{kadelka2023modularity}
A \emph{non-autonomous Boolean network} is defined by $$y(t+1)=H(g(t),y(t)),$$
where $H:\FF_2^{m+n}\rightarrow \FF_2^n$ and $\left(g(t)\right)_{t=0}^\infty$ is a sequence with elements in $\FF_2^m$. We call this type of network non-autonomous because its dynamics will depend on $g(t)$. We use $H^g$ to denote this non-autonomous network.

A state $c\in\FF_2^n$ is a \emph{steady state} of $H^g$ if $H(g(t),c)=c$ for all $t$. Similarly, an ordered set with $r$ elements, $\mathcal{C}=\{c_1,\ldots,c_r\}$, is an \emph{attractor of length $r$} of $H^g$ if $c_2=H(g(1),c_1)$, $c_3=H(g(2),c_2),\ldots, 
c_r=H(g(r-1),c_{r-1})$,
$c_1=H(g(r),c_r)$, $c_2=H(g(r+1),c_1),\ldots$. Note that in general $g(t)$ is not necessarily of period $r$ and may even not be periodic.
\end{dfn}

If $H(g(t),y)=G(y)$ for some network $G$ (that is, it does not depend on $g(t)$)
for all $t$, then $y(t+1)=H(g(t),y(t))=G(y(t))$  and this definition of attractors coincides with the classical definition of attractors for (autonomous) Boolean networks (Definition~\ref{def_attractors}). 

\begin{dfn} 
Consider a controlled non-autonomous network given by $y(t+1)=F_2(g(t),y(t),\mu)$,
where $g(t)$ is a trajectory representation of an attractor $\mathcal{C}_1$ of an upstream network. We say that a control $\mu_2$ \emph{stabilizes} this network, $F_2^{\mathcal{C}_1}$ (as in Definition~\ref{def:non_aut}),  at an attractor $\mathcal{C}_2$ when the resulting network after applying $\mu_2$ (denoted here as $F_2^{\mathcal{C}_1,\mu_2}$) has $\mathcal{C}_2$ as its unique attractor. 
For non-autonomous networks the definition of unique attractor requires that $\left(g(t),y(t)\right)_{t=0}^\infty$ has a unique periodic trajectory up to shifting of $t$ (which is automatically satisfied if $\mathcal{C}_1$ or $\mathcal{C}_2$ is a steady state).
\end{dfn}

\new{
\begin{figure}
    \centering
    \includegraphics[width=\textwidth]{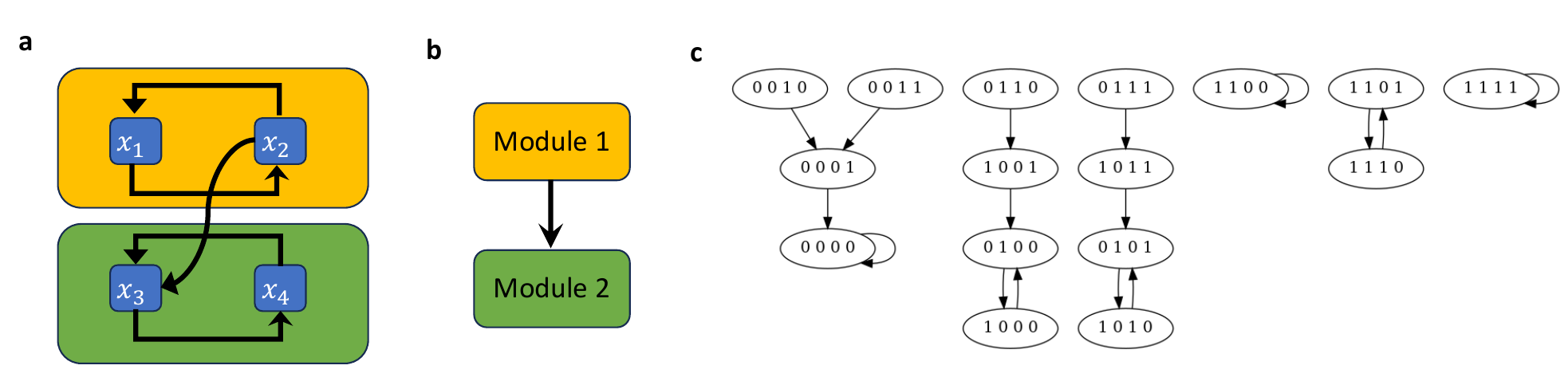}
    \caption{Wiring diagram, directed acyclic graph, and state space for the network in Examples~\ref{ex:control0} and \ref{ex:stabilize}.
    (a) Wiring diagram where the \new[non-trivial]{} modules are highlighted by amber and green boxes. (b) Directed acyclic graph describing the corresponding connections between the \new[non-trivial]{} modules.
    (c) State space of the network generated with Cyclone~\cite{dimitrova2023cyclone}.}
    \label{fig:state-space}
\end{figure}   
}

\section{Control via Modularity and Canalization}
In addition to using the modular structure of the network, we can take advantage of the canalizing structure of the regulatory functions to identify control targets.

We first review some concepts and definitions, and introduce the concept of \emph{canalization}.

\begin{dfn}\label{def_essential}
A Boolean function $f(x_1,\ldots,x_n)$ is \emph{essential} in the variable $x_i$ if there exists an $\vecx \in \{0,1\}^n$ such that 
$$f(\vecx)\neq f(\vecx\new{+\, \mathbf{e_i}}),$$
where \new{$\mathbf{e_i}$} is the $i$th unit vector. In that case, we also say $f$ \emph{depends} on $x_i$.
\end{dfn}

\begin{dfn}\label{def_canalizing}
A Boolean function $f(x_1,\new{\ldots},x_n)$ is \emph{canalizing} if there exists a variable $x_i$, a Boolean function $g(x_1,\ldots,x_{i-1},x_{i+1},\ldots,x_n)$ and $a,\,b\in\{0,\,1\}$ such that
$$f(x_1,x_2,\new{\ldots},x_n)= \begin{cases}
b,& \ \text{if}\ x_i=a\\
g(x_1,x_2,\new{\ldots},x_{i-1},x_{i+1},\new{\ldots},x_n),& \ \text{if}\ x_i\neq a
\end{cases}$$
In that case, we say that $x_i$ \emph{canalizes} $f$ (to $b$) and call $a$ the \emph{canalizing input} (of $x_i$) and $b$ the \emph{canalized output}. 
\end{dfn}

\begin{dfn}\label{def_nested_canalizing}
A Boolean function $f(x_1,\ldots,x_n)$ is \new{a} \emph{nested canalizing \new{function (NCF)}} with respect to the permutation $\sigma \in \mathcal{S}_n$, inputs $a_1,\ldots,a_n$ and outputs $b_1,\ldots,b_n$, if
\begin{equation*}f(x_{1},\ldots,x_{n})=
\left\{\begin{array}[c]{ll}
b_{1} & x_{\sigma(1)} = a_1,\\
b_{2} & x_{\sigma(1)} \neq a_1, x_{\sigma(2)} = a_2,\\
b_{3} & x_{\sigma(1)} \neq a_1, x_{\sigma(2)} \neq a_2, x_{\sigma(3)} = a_3,\\
\vdots  & \vdots\\
b_{n} & x_{\sigma(1)} \neq a_1,\ldots,x_{\sigma(n-1)}\neq a_{n-1}, x_{\sigma(n)} = a_n,\\
1\new{+}\, b_{n} & x_{\sigma(1)} \neq a_1,\ldots,x_{\sigma(n-1)}\neq a_{n-1}, x_{\sigma(n)} \neq a_n.
\end{array}\right.\end{equation*}
The last line ensures that $f$ actually depends on all $n$ variables.
\end{dfn}

To account for partial canalization, we also define $k$-canalizing functions\new{,} which were first introduced in~\cite{he2016stratification}.

\begin{dfn}\label{def:kcanalizing}
A Boolean function $f(x_1,\ldots,x_n)$ is \emph{$k$-canalizing}, where $1 \leq k \leq n$, with respect to the permutation $\sigma \in \mathcal{S}_n$, inputs $a_1,\ldots,a_k$, and outputs $b_1,\ldots,b_k$ if
\begin{equation*}
\begin{array}{l}
f(x_{1},\ldots,x_{n})=\\
\left\{\begin{array}[c]{ll}
b_{1} & x_{\sigma(1)} = a_1,\\
b_{2} & x_{\sigma(1)} \neq a_1, x_{\sigma(2)} = a_2,\\
b_{3} & x_{\sigma(1)} \neq a_1, x_{\sigma(2)} \neq a_2, x_{\sigma(3)} = a_3,\\
\vdots  & \vdots\\
b_{k} & x_{\sigma(1)} \neq a_1,\ldots,x_{\sigma(k-1)}\neq a_{k-1}, x_{\sigma(k)} = a_k,\\
f_C\not\equiv b_k & x_{\sigma(1)} \neq a_1,\ldots,x_{\sigma(k-1)}\neq a_{k-1}, x_{\sigma(k)} \neq a_k,
\end{array}\right.
\end{array}
\end{equation*}
where $f_C = f_C(x_{\sigma(k+1)},\ldots,x_{\sigma(n)})$ is the \emph{core function}, a Boolean function on $n-k$ variables. When $f_C$ is not canalizing, then the integer $k$ is the \emph{canalizing depth} of $f$ \cite{Layne:2012aa}. Note that an $n$-canalizing function (i.e., a function where all variables become eventually canalizing) is \new[also called a \emph{nested canalizing function} (NCF)]{an NCF}. 
\end{dfn}

We restate the following stratification theorem for reference.

\begin{thm}[\cite{he2016stratification}]\label{thm:he}
Every Boolean function $f(x_1,\ldots,x_n)\not\equiv 0$ can be uniquely written as 
\begin{equation}\label{eq:matts_theorem}
    f(x_1,\ldots,x_n) = M_1(M_2(\cdots (M_{r-1}(M_rp_C + 1) + 1)\cdots)+ 1)+ q,
\end{equation}
where each $M_i = \displaystyle\prod_{j=1}^{k_i} (x_{i_j} + a_{i_j})$ \new{is a nonconstant function, called an \emph{extended monomial}}, $p_C$ is the \emph{core polynomial} of $f$, and $k = \displaystyle\sum_{i=1}^r k_i$ is the canalizing depth. Each $x_i$ appears in exactly one of $\{M_1,\ldots,M_r,p_C\}$, and the only restrictions are the following ``exceptional cases'':
\begin{enumerate}
    \item If $p_C\equiv 1$ and $r\neq 1$, then $k_r\geq 2$;
    \item If $p_C\equiv 1$ and $r = 1$ and $k_1=1$, then $q=0$.
\end{enumerate}
When $f$ is not canalizing (\textit{i.e.}, when $k=0$), we simply have $p_C = f$.
\end{thm}

\begin{dfn}\label{def_layers}
Given a Boolean function $f(x_1,\ldots,x_n)$ represented as in Equation~\ref{eq:matts_theorem}, we call the extended monomials $M_i$ the \emph{layers} of $f$ and \new{if $i<j$, we say that $M_i$ and its variables are \emph{more dominant} and $M_j$ and its variables are \emph{less dominant}.} We also define, as in~\cite{kadelka2017influence}, the \emph{layer structure} as the vector $(k_1,\ldots,k_r)$, which describes the number of variables in each layer.  
\end{dfn}
Note that $f$ is nested canalizing if and only if $k_1+\cdots+k_r = n$. 

\begin{remark}~\label{rem:layer_output} Here we note the following important properties of layers of canalization.
\begin{enumerate}
    \item[(a)] Theorem~\ref{thm:he} shows that any Boolean function has a unique extended monomial form given by Equation~\ref{eq:matts_theorem}, in which the variables are partitioned into different layers based on their dominance. Any variable that is canalizing (independent of the values of other variables) is in the first layer. Any variable that ``becomes'' canalizing when excluding all variables from the first layer is in the second layer, etc. All remaining variables that never become canalizing are part of the core polynomial. The number of variables that eventually become canalizing is the canalizing depth of the function. NCFs are exactly those functions where all variables become eventually canalizing (note not all variables of an NCF must be in the first layer).
    \item[(b)] While variables in the same layer may have different canalizing input values, they all share the same canalized output value, \textit{i.e.}, they all canalize a function to the same output. On the other hand, the outputs of two consecutive layers are distinct. Therefore, the number of layers of a $k$-canalizing function, expressed as in Definition~\ref{def:kcanalizing}, is simply one plus the number of changes in the vector of canalized outputs, $(b_1,b_2,\ldots,b_k)$.
\end{enumerate}
\end{remark} 

\begin{eg}
The Boolean functions 
$$f(x_1,x_2,x_3,x_4)=x_1\wedge (\neg x_2\vee (x_3 \wedge x_4)) \qquad \text{and} \qquad g(x_1,x_2,x_3,x_4)=x_1\wedge (\neg x_2\vee x_3 \vee x_4)$$ 
are nested canalizing. The function $f$ has layer structure $(1,1,2)$ because
\begin{itemize}
\item $x_1$ canalizes $f$ to $0$ if it receives its canalizing input $0$;
\item if this is not the case, $x_2=0$ canalizes $f$ to $1$;
\item $x_3=0$ or $x_4=0$ canalizes $f$ to $0$.
\end{itemize}
On the other hand, $g$ has layer structure $(1,3)$. As for $f$, $x_1=0$ can\new{a}lizes $g$ to $0$. If this does not happen, any of the following, $x_2=0, x_3=1, x_4=1$, canalizes $g$ to $1$.
\end{eg}

While finding the layer structure of a Boolean function is an \textbf{NP}-hard problem, there exist several algorithmic implementations~\cite{dimitrova2022revealing}.

\new{Next we will define phenotypes.} In meaningful biological networks, the attractors correspond to phenotypes. Typically a small subset of all Boolean variables is used to define phenotypes. 

\begin{dfn}
Given a Boolean network $F$ with attractor space $\mathcal{A}(F)$ and phenotype-defining variables $\mathcal P \subset\{x_1,\ldots,x_n\}$, we associate the same phenotype to all attractors $\mathcal C \in \mathcal{A}(F)$ that are identical in $\mathcal P$. The states in $\mathcal P$ will be called markers of the phenotype.
\end{dfn}

Suppose $F = F_1\rtimes_P F_2$ is a decomposable network, and that there is a phenotype
that depends on variables in $F_2$ only (that is, all markers of the phenotype are part of $F_2$), and that we wish to control the phenotype through $F_2$.
The most straightforward approach is to set the variables that the phenotype depends on to the appropriate values that result in the desired phenotype. However, such intervention may not be experimentally possible. Instead, we can exploit the canalizing properties of the functions corresponding to the nodes connecting the modules $F_1$ and $F_2$ to identify control targets.

By Theorem~\ref{thm:he}, the variables of any Boolean update function  can be ordered by importance/dominance, based on  which layer they appear in, \new{which is why in Definition~\ref{def_layers} we called variables in lower/upper layers ``more/less dominant," respectively.}
Thus, once we control a variable in a certain layer (by setting it to its canalizing input value), any further control of variables in less dominant layers will have no effect on the function (and thus on the network). We state this fact in the following lemma.

\begin{lem}\label{lem:ctrl+canal}
Suppose $F = F_1\rtimes_P F_2$ is a decomposable network. Suppose further that only one node $x \in F_2$ with update function $f_x$ is regulated by nodes in $F_1$. If $f_x$ is canalizing with $r$ layers, let \new{$M_{\ell}$} be the most dominant layer of $f_x$ which contains nodes from $F_1$. If all regulators of $x$ from $F_1$ appear in the core polynomial, we set $\ell=r+1$. Then, setting $y\not\in F_1$ to its canalizing value decouples the systems $F_1$ and $F_2$, as long as $y$ appears in a layer \new{of $f_x$ that is more dominant than $M_{\ell}$}.
\end{lem}

\begin{proof}
The lemma is a direct consequence of Theorem~\ref{thm:he}. If $y$ receives its canalizing input and is in a more \new{dominant} layer of $f_x$ than all variables in $F_1$, then none of these variables can affect $f_x$ anymore. Thus, controlling $y$ to receive its canalizing input eliminates the link between $F_1$ and $F_2$.
\end{proof}

The modularity approach in \cite{kadelka2023modularity} yields an acyclic directed graph after one collapses each module into a single node. This endows a natural partial ordering on the collection modules of a network where one module precedes the next if there is a path for every node in the first module which ends in the second module. While not all modules are comparable, we can speak of chains of modules which consist of subsets of the partial ordering which are totally ordered.
Furthermore, one can rank the modules based on the percentile scores (i.e., rank module $k$ out of $m$ modules). 
This type of ranking has been studied in~\cite{plaugher2024pancreatic}, where it was shown 
that the importance of the modules is strongly correlated with the aggressiveness of mutations occurring within those modules and the effectiveness of interventions.

\begin{thm}\label{thm:ctrl+canal}
    Suppose $F = F_1\rtimes_{P_1} F_2\rtimes_{P_2}\cdots\rtimes_{P_{n-1}} F_m$ is a decomposable network. If for some $i<j$,
    \begin{enumerate}[(i)]
    \item only one node $x \in F_j$ with update function $f_x$ is regulated by nodes in $F_i$, and
    \item $f_x$ is a canalizing function, which possesses none of the variables from $F_i$ in its most \new{dominant} layer, and
    \item the phenotype of interest only depends on variables in $F_j$ as well as modules $F_k$, for which any chain containing $F_i$ and $F_k$ also contains $F_j.$ \new[(see Figure~\ref{fig:thm412explained} for an example)]{}
    
    \end{enumerate}
    then the module $F_i$ can be excluded from the control search by setting any node $y \notin F_i$ to its canalizing input, as long as this node appears in a more dominant layer of $f_x$ than all inputs of $f_x$ that are part of $F_i$.
\end{thm}
\begin{proof}
By Lemma~\ref{lem:ctrl+canal}, setting $y$ to its canalizing value results in decoupling $F_i$ and $F_j$. $F_i$ will no longer have any effect on $F_j$, and thus, due to condition (iii), on the phenotype of interest. \new{Therefore}, $F_i$ can be removed from the control search.
\end{proof}

\begin{figure}
    \centering
    \includegraphics[width=0.6\textwidth]{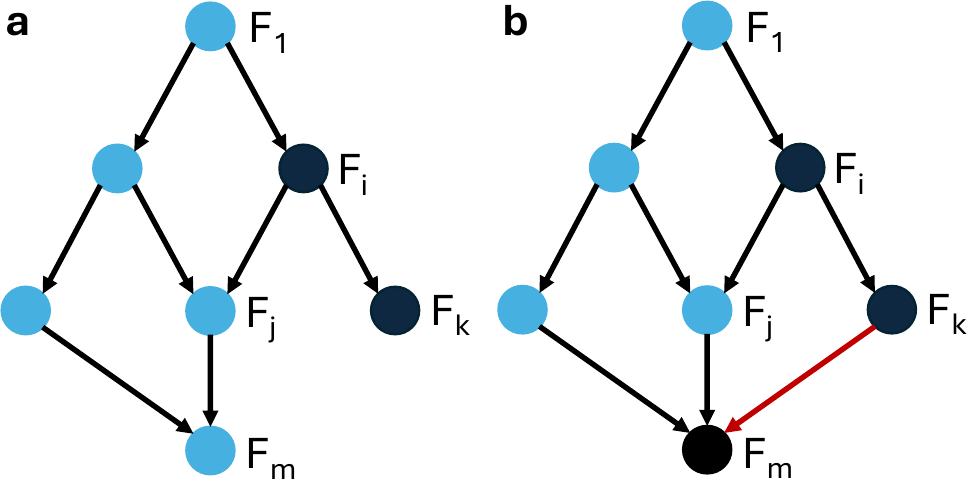}
    \caption{Example of a modular directed acylic graph structure to illustrate Condition (iii) in Theorems~\ref{thm:ctrl+canal} and \ref{thm:edge_control}. (a) Module $F_i$ can be removed from the control search as long as conditions (i) and (ii) are satisfied, and the phenotype of interest depends only on any subset of variables that are part of the blue modules. (b) When a node in module $F_m$ is regulated by a node in module $F_k$ (indicated by the red edge in the directed acyclic graph),  the phenotype may no longer depend on nodes in $F_m$, in order for module $F_i$ to be removable from the control search.}
    \label{fig:thm412explained}
\end{figure}

\new{The third condition of this theorem is illustrated in Figure~\ref{fig:thm412explained}.}

\new{\begin{remark}
    The method in Theorem~\ref{thm:ctrl+canal} can be extended to the case when $F_i$ and $F_j$ are connected via multiple nodes. In that case\new{,} decoupling is achieved through the same procedure presented above, applied to each node in $F_j$ that is regulated by nodes in $F_i$.
\end{remark}
}

\begin{figure}
    \centering
    \includegraphics[width=0.6\textwidth]{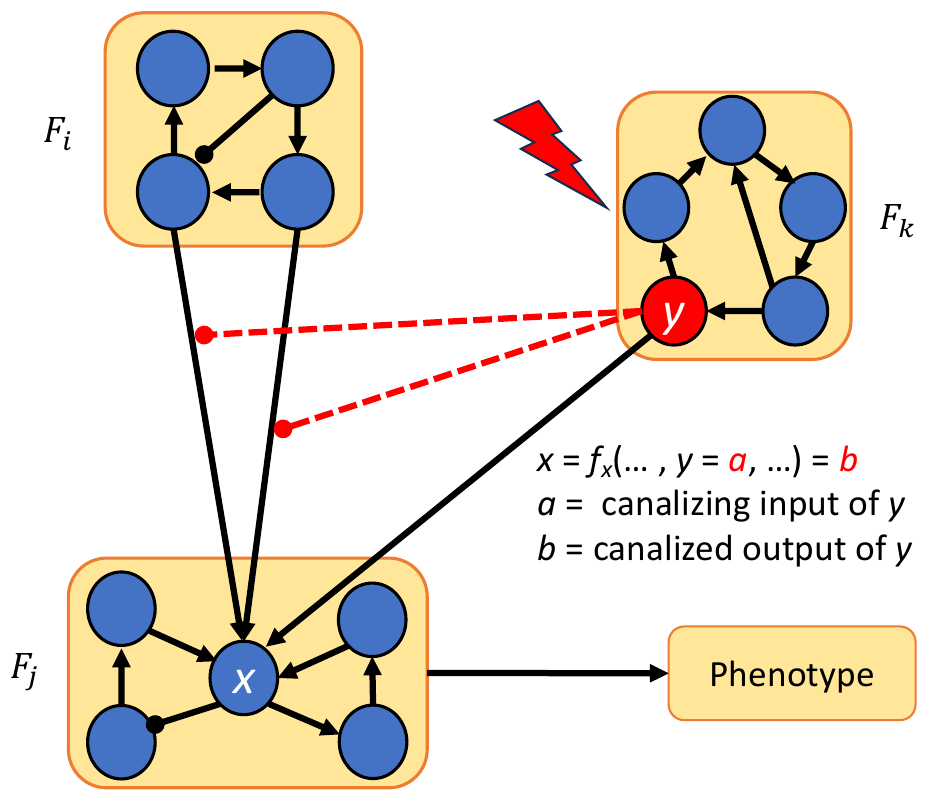}
    \caption{Control via modularity and canalization. Once the network is decomposed into modules $F_1, \new{\ldots}, F_m$, we can override the effect of module $F_i$ by the using another module ($F_k$ in this case) whose variables are inputs of $f_x$ that are \new{located in more dominant layers} than the layers containing the variables of $F_i$. 
    }
    \label{fig:control_with_canalization}
\end{figure}

Theorem~\ref{thm:ctrl+canal} is illustrated in Figure~\ref{fig:control_with_canalization}. Note that node $y$ can be in $F_j$ or some other module as in the figure. In this theorem, we assumed that none of the variables of $F_i$ are in the most dominant layer in the update rules of variables in $F_j$.
If some variables of $F_i$ are in the most dominant layer, we can still remove module $F_i$ from the control search using an edge control, as shown in the following theorem. 
\begin{thm}\label{thm:edge_control}
    Suppose $F = F_1\rtimes_{P_1} F_2\rtimes_{P_2}\cdots\rtimes_{P_{n-1}} F_m$ is a decomposable network. If for some $i<j$,
    \begin{enumerate}[(i)]
    \item only one node $x \in F_j$ with update function $f_x$ is regulated by nodes in $F_i$, and \item $f_x$ is a canalizing function with some variables from $F_i$ in its most \new{dominant} layer, and
    \item the phenotype of interest only depends on variables in $F_j$ as well as modules $F_k$, for which any chain containing $F_i$ and $F_k$ also contains $F_j.$ \new[(see Figure~\ref{fig:thm412explained} for an example)]{}
    \end{enumerate}
    then the module $F_i$ can be excluded from the control search by applying an edge control to any input in the most dominant layer of $f_x$. 
\end{thm}
\begin{proof}
Let $y \in F_i$ such that \new{$y$ appears as a variable in $f_x$}, and that $y$ is located in the most dominant layer $f_x$. 
Then, 
setting $y$ to its canalizing value results in decoupling the subnetworks $F_i$ and $F_j$. Thus, $F_i$ will no longer have any effect on $F_j$ and thus it can be removed from the control search.   
\end{proof}

\begin{remark}
    The method can be extended to the case when $F_i$ and $F_j$ are connected via multiple nodes (that is, condition (i) in Theorem~\ref{thm:ctrl+canal} and Theorem~\ref{thm:edge_control} can be relaxed. In that case decoupling is achieved through the same procedure presented above, applied to each node in $F_j$ with regulators from $F_i$. 

    We further note that condition (ii) in the theorems above is generally very restrictive as only a small proportion of Boolean functions in $n>3$ variables \new{are} canalizing, \new{let alone} nested canalizing. However, as shown in~\cite{kadelka2024meta}, most biological Boolean network models are almost entirely governed by nested canalizing functions.
\end{remark}

\section{\new{An application: Control of a Blood Cancer Boolean Model}}

To showcase these methods, we will now decompose a published Boolean network model into its modules, and then identify the minimal set of controls for the entire network by exploiting the canalizing structure of the regulatory functions within the modules. The identified set of controls will force the entire system into a desired attractor.

\begin{figure}
    \centering
    \includegraphics[width=\textwidth]{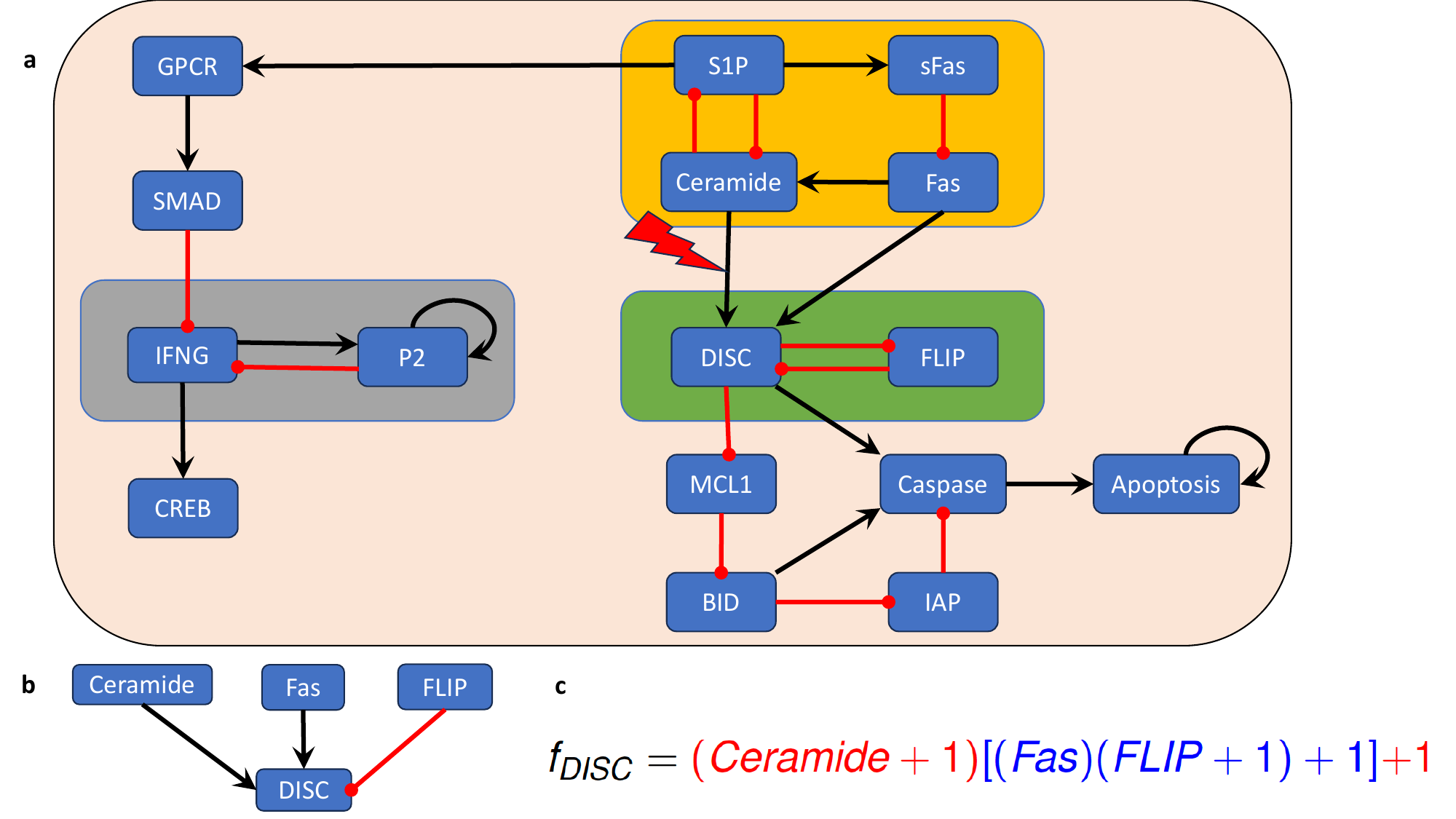}
    \caption{(a) Wiring diagram of the T-LGL model, published in \cite{Saadatpour:2011aa}, which describes the mechanisms that regulate apoptosis. The non-trivial modules \new{(i.e., the modules containing multiple nodes)} are highlighted by amber, green, and gray boxes.
    (b) The regulatory inputs of the node DISC. (c) Writing the regulatory function corresponding to node DISC in its \new{extended} monomial form (Theorem~\ref{thm:he}) reveals its canalizing structure.}
    \label{fig:tlgl}
\end{figure}    
    
    We consider a Boolean network model for the blood cancer large granular lymphocyte (T-LGL) leukemia, which was published in \cite{Saadatpour:2011aa}. T-LGL leukemia is a clonal hematological disorder characterized by persistent increases of large granular lymphocytes in the absence of reactive cause~\cite{rashid2014t}. The wiring diagram of this model is depicted in Figure~\ref{fig:tlgl}a. This network has 16 nodes and three non-trivial modules (\new{i.e., modules containing more than one node,} highlighted by the amber, green, and gray boxes in Figure~\ref{fig:tlgl}a).
   The control objective here is to identify control targets that lead the system to programmed cell death. In other words, we aim to direct the system into an attractor that has the marker apoptosis ON.   

Since the model has three non-trivial modules, the approach described in Section~\ref{sec:control} would require us to identify control targets for three modules. However, an exploitation of the canalizing structure and common sense reveals that we do not need to control every module to ensure apoptosis, the desired control objective. First, irrespective of canalization, the module highlighted in gray in Figure~\ref{fig:tlgl}a does not affect the phenotype apoptosis. Therefore, we can focus on the modules ``upstream" of apoptosis (i.e., the amber and green modules in Figure~\ref{fig:tlgl}a).

In this case, we will apply Theorem~\ref{thm:edge_control} to identify control targets for this model. 
We note that the edges from the upstream module (amber box in Figure~\ref{fig:tlgl}a) to the downstream module (green box in Figure~\ref{fig:tlgl}a) all end in the node DISC. Therefore, we will investigate the canalizing properties of the regulatory function of DISC (see Figure~\ref{fig:tlgl}b),  
   \[
   f_{DISC}=Ceramide\lor(Fas\land \overline{FLIP}). 
    \]
   Using the approach described in~\cite{dimitrova2022revealing}, we find that $f_{DISC}$ has two canalizing layers, $L_1=\{Ceramide\}$ and $L_2=\{Fas, FLIP\}$, and \new{extended monomial form (Theorem~\ref{thm:he}) given in Figure~\ref{fig:tlgl}c).}
 
We note that the only variable in the most dominant canalizing layer, $Ceramide$, is in the upstream module. Thus, we can decouple the modules via an edge control on the connection between the upstream and downstream modules. That is, the constant expression of the edge from Ceramide to DISC will decouple the two modules and will lead to constant expression of DISC. We can check that this control is effective at stabilizing the system in the desired attractor and the control set obviously has minimal size.

In summary, in this example we used an edge control to decouple the upstream and downstream modules and then identified a control target in the downstream module which contains the markers of the phenotype of interest.

\section{Conclusion}
Model-based control is a mainstay of industrial engineering, and there is a well-developed mathematical theory of optimal control that can be applied to models consisting of systems of ordinary differential equations. While this model type is also commonly  used in biology, for instance in biochemical network modeling or epidemiology and ecology, there are many biological systems that are more suitably modeled in other ways. Boolean network models provide a way to encode regulatory rules in networks that can be used to capture qualitative properties of biological networks, when it is unfeasible or unnecessary to determine kinetic information. While they are intuitive to build, they have the drawback that there is very little mathematical theory available that can be used for model analysis, beyond simulation approaches. And for large networks, simulation quickly becomes ineffective. 

The results in this paper, building on those in \cite{kadelka2023modularity}, can be considered as a contribution to a mathematical control theory for Boolean networks, incorporating key features of biological networks. There are many open problems that remain, and we hope that this work will inspire additional developments. 

Our concrete contributions here are as follows. The modularization method makes the control search far more efficient and allows us to combine controls at the module level obtained with different control methods. For example, methods based on computational algebra~\cite{murrugarra2016identification,sordo2020control} can identify controllers that can create new (desired) steady states, which other methods cannot. Feedback vertex set~\cite{mochizuki2013dynamics,zanudo2017structure} is a structure-based method that identifies a subset of nodes whose removal makes the graph acyclic. Stable motifs~\cite{zanudo2015cell} are based on identifying strongly connected subgraphs in the extended graph representation of the Boolean network. Other control methods include~\cite{choo2018phenotype,LauraCF2022,borriello2021basis}. We can use any combination of these methods to identify the controls in each module.

\section{Acknowledgments}
Author Matthew Wheeler was supported by The American Association of Immunologists through an Intersect Fellowship for Computational Scientists and Immunologists. This work was further supported by the Simons foundation [grant numbers 712537 (to C.K.), 850896 (to D.M.), 516088 (to A.V.)]; the American Mathematical Society and the Simons Foundation [Enhancement Grant for PUI faculty (to A.V.)]; the National Institute of Health [grant number 1 R01 HL169974-01 (to R.L.)];  and the Defense Advanced Research Projects Agency [grant number HR00112220038 (to R.L.)]. The authors also thank the Banff International Research Station for support through its Focused Research Group program during the week of May 29, 2022 (22frg001), which was of great help in framing initial ideas of this paper.

\bibliographystyle{unsrt}
\bibliography{control_ref}

\end{document}